\documentclass{llncs}
\usepackage{llncsdoc}
\usepackage{supertabular}
\usepackage{cite}
\usepackage{graphicx}
\usepackage{amsmath,amsfonts,amssymb}
\usepackage[ruled,vlined,linesnumbered]{algorithm2e}
\usepackage{lineno,hyperref}
\usepackage{multirow}
\usepackage{floatrow}
\newfloatcommand{capbtabbox}{table}[][\FBwidth]

\newcommand{\Jac}{\mathcal{J}}
\newcommand{\rank}{\mathrm{rank}}

\begin{document}

\title{{A Special Homotopy Continuation Method For A Class of Polynomial Systems}
\thanks{The work is partly supported by the projects NSFC Grants 11471307,  11290141, 11271034 and
61532019.}}

\author{Yu Wang\inst{1} \and Wenyuan Wu\inst{2}$\thanks{Corresponding author.}$ \and Bican Xia\inst{1} }
\institute{LMAM \& School of Mathematical Sciences, Peking University \\ \email{yuxiaowang@pku.edu.cn, xbc@math.pku.edu.cn}\\ \and Chongqing Inst. of Green and Intelligent Techn.\\Chinese Academy of Sciences\\ \email{wuwenyuan@cigit.ac.cn}
}

\date{}
\maketitle
\begin{abstract}
{A special homotopy continuation method, as a combination of the polyhedral homotopy and the linear product homotopy, is proposed for computing all the isolated solutions to a special class of polynomial systems. The root number bound of this method is between the total degree bound and the mixed volume bound and can be easily computed. The new algorithm has been implemented as a program called \texttt{LPH} using C++. Our experiments show its efficiency compared to the polyhedral or other homotopies on such systems. As an application, the algorithm can be used to find witness points on each connected component of a real variety.}
\end{abstract}




\section{Introduction}

In many applications in science, engineering, and economics, solving systems of polynomial equations has been a subject of great importance. 
The homotopy continuation method was developed in 1970s \cite{Garcia1979}\cite{Drexler1977} and has been greatly expanded and developed by many reseachers 
(see for example \cite{Sommese96numericalalgebraic}\cite{Allgower2003Introduction}\cite{Li2003209}\cite{Sommese2005The}\cite{Morgan:2009:SPS:1717962}). Nowadays, homotopy continuation method has become one of the most reliable and efficient classes of numerical methods for finding the isolated solutions to a polynomial system and the so-called {\em numerical algebraic geometry} based on homotopy continuation method has been a blossoming area. There are many famous software packages implementing different homotopy methods, including Bertini\cite{Bates:2013:NSP:2568129}, Hom4PS-2.0\cite{Lee2008}, HOMPACK\cite{Morgan:1989:FIS:63522.64124}, PHCpack\cite{Verschelde:1999:APG:317275.317286}, etc.


Classical homotopy methods compute solutions in complex spaces, while in applications, it is quite common that only real solutions have physical meaning. Computing real roots of an algebraic system is a difficult and fundamental problem in real algebraic geometry. In the field of symbolic computation, there are some famous algorithms dealing with this problem. The cylindrical algebraic decomposition algorithm \cite{cad} is the first complete algorithm which has been implemented and used successfully to solve many real problems. However, in the worst case, its complexity is of doubly exponential in the number of variables. Based on the ideas of Seidenberg \cite{10.2307/1969640} and others, some algorithms for computing at least one point on each connected component of an real algebraic set were proposed through developing the formulation of critical points and the notion of polar varieties, see \cite{ROUILLIER2000716}
\cite{SafeyElDin:2003:PVC:860854.860901}\cite{SafeyElDIn:2016:CPC:2930889.2930929} and references therein. The idea behind is studying an objective function (or map) that reaches at least one local extremum on each connected component of a real algebraic set. 
For example, the function of square of the Euclidean distance to a randomly chosen point was used in 
\cite{Bank2004}\cite{BANK2005377}. On the other hand, some homotopy based algorithms for real solving have been proposed in \cite{Li1993Solving}\cite{Lu06findingall}\cite{Bates:2011:KCR:2134316.2134319}\cite{Besana:2013:CDA:2509858.2509861}\cite{Hauenstein2013}\cite{Shen2014}. For example, in \cite{Hauenstein2013}, a numerical homotopy method to find the extremum of Euclidean distance to a point as the objective function was presented. More recently, the Euclidean distance to a plane was proposed as a linear objective function in \cite{Wu:2013:FPR:2465506.2465954}.


In this paper, we follow the work of \cite{Wu:2013:FPR:2465506.2465954}\cite{Wu2017} to extend complex homotopy methods to finding witness points on the irreducible components of real varieties. To obtain such witness points, we first need to solve a special class of polynomial systems. Combining the polyhedral homotopy and the linear product homotopy, we give a special homotopy method for solving the system of that type. The root number bound of this method is not only easy to compute but also much smaller than the total degree bound and close to the BKK bound \cite{Bernshtein1975} when the polynomials defining the algebraic set is not very sparse. This key observation enables us to design an efficient homotopy procedure to obtain critical points numerically. The ideas and algorithms we proposed in this article avoid a great number of divergent paths to track compared with the total degree homotopy and save the great time cost for mixed volume computation compared with the polyhedral homotopy. The new algorithm has been implemented as a program called \texttt{LPH} using C++. Our experiments show its efficiency compared to the polyhedral or other homotopies on such systems.

The rest of this paper is organized as follows. Section 2 describes some preliminary concepts and results. Section 3 introduces a special type of polynomial systems we are considering. The new homotopy for these polynomial systems is also presented. It naturally leads to an algorithm which is described in Section 4. Based on this algorithm, in Section 5, we present a method to find real witness points of positive dimensional varieties, together with an illustrative example. The experimental performance of the software package \texttt{LPH}, which is an implementation of the method in C++, is given in Section 6. 

\section{Preliminary}

\subsection{Algebraic Sets and Genericity}

For a polynomial system $f:{\bbbc^n} \to {\bbbc^k}$, let $V(f) = \{ x \in {\bbbc^n}|f(x) = 0\} $ and ${V_\bbbr}(f) = V(f) \cap {\bbbr^n} = \{ x \in {\bbbr^n}|f(x) = 0\}$ be the set of complex solutions and the set of real solutions of $f(x) = 0$, respectively.  A set $X \subset {\bbbc^n}$ is called an algebraic set if $X = V(g)$, for some polynomial system $g$.

An algebraic set $X$ is irreducible if there does not exist a decomposition ${X_1} \cup {X_2} = X$ with ${X_1},{X_2} \ne X$ of $X$ as a union of two strict algebraic subsets. An algebraic set is reducible, if there exist such decomposition. For example, the algebraic set $V(xy) \subset {\bbbc^2}$ is consisting of the two coordinate axes, and is obviously the union of $V(x)$ and $V(y)$, hence reducible.

For an irreducible algebraic set $X$, the subset of smooth (or manifold) points ${X_{reg}}$ is dense, open and path connected (up to the Zariski topology) in $X$. The dimension of an irreducible algebraic set $X$ is the dimension of ${X_{reg}}$ as a complex manifold.

Let $\Jac_f(x)$ denote the $n \times k$ Jacobian matrix of $f$ evaluated at $x$. By the Implicit Function Theorem, for an irreducible algebraic set $X$ defined by a reduced system $f$, $x \in {X_{reg}} \Leftrightarrow \rank(\Jac_f(x)) = n - \dim X$. When $n = k$, the system $f$ is said to be a square system. In this case, a point $x \in V(f)$ is nonsingular if $\det (\Jac_f(x)) \ne 0$, and singular otherwise.

On irreducible algebraic set, we can define the notion of genericity, adapted from\cite{Sommese2005The}.

\begin{definition}
Let $X$ be an irreducible algebraic set. Property {\tt P} holds generically on $X$, if the set of points in $X$ that do not satisfy property {\tt P} are contained in a proper algebraic subset $Y$ of $X$. The points in $Y$ are called {\em nongeneric points}, and their complements $X\backslash Y$ are called {\em generic points}.
\end{definition}

\begin{remark} From the definition, one sees that the notion of generic is only meaningful in the context of property {\tt P} in question.
\end{remark}

Every algebraic set $X$ has a (uniquely up to reordering) expression $X = {X_1} \cup  \ldots \cup {X_r}$ with ${X_i}$ irreducible and ${X_i} \not\subset {X_j}$ for $i \ne j$. And ${X_i}$ are the irreducible components of $X$. The dimension of an algebraic set is defined to be the maximum dimension of its irreducible components. An algebraic set is said to be pure-dimensional if each of its components has the same dimension.

\subsection{Trackable Paths} In homotopy continuation methods, the notion of path tracking is fundamental, the following definition of trackable solution path is adapted from \cite{MR2728983}.

\begin{definition} Let $H(x,t):{\bbbc^n} \times \bbbc \to {\bbbc^n}$ be polynomial in $x$ and complex analytic in $t$, and let ${x^*}$ be nonsingular isolated solution of $H(x,0) = 0$, we say ${x^*}$ is {\em trackable} for $t \in [0,1)$ from 0 to 1 using $H(x,t)$ if there is a smooth map ${\xi _{{x^*}}}:[0,1) \to {\bbbc^n}$ such that ${\xi _{{x^*}}}(0) = {x^*}$, and for $t \in [0,1)$, ${\xi _{{x^*}}}(t)$ is a nonsingular isolated solution of $H(x,t) = 0$. The solution path started at ${x^*}$is said to be {\em convergent} if $\mathop {\lim }\limits_{t \to 1} {\xi _{{x^*}}}(t) \in {\bbbc^n}$, and the limit is called the {\em endpoint} of the path.
\end{definition}

\subsection{Witness Set and Degree of an Algebraic Set}

Let $X \subset {\bbbc^n}$ be a pure $i$-dimensional algebraic set, given a generic co-dimension $i$ affine linear subspace $L \subset {\bbbc^n}$, then $W = L \cap X$ consists of a well-defined number $d$ of points lying in ${X_{reg}}$. The number $d$ is called the degree of $X$ and denoted by $\deg (X)$. We refer to $W$ as a set of witness points of $X$, and call $L$ the associated $(n - i)$-slicing plane, or slicing plane for short \cite{Sommese2005The}.

It will be convenient to use the notations adapted from (\cite{Sommese2005The},Chapter 8), when we prove the theorems in Section 3.

\begin{enumerate}
\item Let $\left\langle {{e_1}, \ldots, {e_n}} \right\rangle $ be the $n$ dimensional vector space having basis elements ${e_1}, \ldots , {e_n}$ with complex coefficients. That is, a point in this space may be written as $\sum\limits_{i = 1}^n {{c_i}{e_i}} $, with ${c_i} \in \bbbc$ for $i = 1, \ldots, n$. Note that we have not specified anything about the basis elements, it could be individual variables, monomials, or polynomials.
\item Let $\left\{ {{p_1}, \ldots, {p_n}} \right\} \otimes \left\{ {{q_1}, \ldots, {q_m}} \right\}$ be the product of two sets, that is, the set $\left\{ {{p_i} \cdot {q_j}|i = 1, \ldots, n;j = 1, \ldots, m} \right\}$. In Section 3 we take this product as the image inside the ring of polynomials; that is, $x \otimes y = xy$ is just the product of two polynomials.
\item Define $P \times Q = \{ pq|p \in P,q \in Q\} $. Accordingly, we have $\left\langle P \right\rangle  \times \left\langle Q \right\rangle  \subset \left\langle {P \otimes Q} \right\rangle $.
\item For repeated products, we use the shorthand notations ${P^{(2)}} = P \otimes P$,${\left\langle P \right\rangle ^{(2)}} = \left\langle P \right\rangle  \times \left\langle P \right\rangle $, and similar for three or more products.
\item For a square polynomial system $P$, we denote by $MV(P)$ the mixed volume of the system $P$.
\end{enumerate}

\subsection{Critical Points}

Let $X \subset {\bbbc^n}$ be an algebraic set defined by a reduced polynomial system $f = \{ {f_1}, \ldots, {f_k}\} $, and objective function $\Phi $ is polynomial function restricted to $X$.

\begin{definition}  A point $x \in X$ is a {\em critical point} of $\Phi $ if and only if $x \in {X_{reg}}$ and $\rank[\nabla \Phi {(x)^T},\Jac_f(x)] = \rank[\nabla \Phi {(x)^T},\nabla f_1^T,...,\nabla f_k^T] \leqslant k$, where $\nabla \Phi (x)$ is the gradient vector of $\Phi $ evaluated at $x$.
\end{definition}

Let $Y$ denote the zero dimensional critical sets of $\Phi $. One way to compute the critical points is to introduce auxiliary unknowns and consider a zero dimensional variety $\hat Y$ and then project $\hat Y$ onto $Y$. We use Lagrange Multipliers to define a squared system as follows
\begin{equation}\label{eq:1}
F(x,\lambda ): = \left[ {\begin{array}{*{20}{c}}
  f \\
  {{\lambda _0}\nabla \Phi {{(x)}^T} + {\lambda _1}\nabla f_1^T + ... + {\lambda _k}\nabla f_k^T}
\end{array}} \right]
\end{equation}
Note that if ${x^*} \in X$ is a critical point of $\Phi $, then there exist ${\lambda ^*} \in {\bbbp^k}$, such that $F({x^*},{\lambda ^*}) = 0$ by the Fritz John condition\cite{John2014Extremum}. In the affine patch where ${\lambda _0} = 1$,  the system $F$ becomes a square system, and its solution $({x^*},{\lambda ^*})$ projects to critical point ${x^*}$. We will use system (\ref{eq:1}) in Section 5 with an objective function $\Phi$ defined by a linear function, and consider the affine patch where ${\lambda _0} = 1$, to find at least one point on each component of ${V_\bbbr}(f)$.

\section{Main Idea}

In this section, we give a description of our idea. First we introduce a family of polynomial equations that we will be considering.

We consider the following class of polynomial systems:
\begin{equation}\label{eq:2}
F(x,\lambda ) = \left\{ \begin{gathered}
  f \hfill \\
  J \cdot \lambda  - \beta \hfill \\
\end{gathered}  \right.
\end{equation}
where
\begin{enumerate}
\item $f = \left\{ {{f_1}, \ldots, {f_k}} \right\}$ are polynomials in $\bbbc\left[ {{x_1}, \ldots, {x_n}} \right]$, and $V({f_1}, \ldots, {f_k})$ is a pure $n-k$ dimension algebraic set in ${\bbbc^n}$.
\item $J = \left( {\begin{array}{*{20}{c}}
  {{g_{11}}}& \cdots &{{g_{1k}}} \\
   \vdots & \ddots & \vdots  \\
  {{g_{n1}}}& \cdots &{{g_{nk}}}
\end{array}} \right)$ and ${g_{ij}} (1 \leqslant i \leqslant n,1 \leqslant j \leqslant k)$ are polynomials in $\bbbc\left[ {{x_1}, \ldots, {x_n}} \right]$ with $\mathop {\max }\limits_{i,j} \deg ({g_{ij}}) = d$.
\item $\beta = {({\beta_1}, \ldots, {\beta_n})^{\rm T}}$ is a nonzero constant vector in ${\bbbc^n}$, $\lambda=({\lambda _1}, \ldots, {\lambda _k})^{\rm T}$ are unknowns, and $n > k \geqslant 1$.
\end{enumerate}


\begin{remark} \label{re:2}
Note that, for any invertible $n \times n$ matrix $A$, $F(x,\lambda ) = \left\{ {f,J \cdot \lambda  - \beta} \right\}$ and $F'(x,\lambda ) = \left\{ {f,A\cdot(J\cdot\lambda  - \beta)} \right\}$ have the same solutions. It's easy to know that there exists an invertible matrix $A$ such that $A \cdot \beta = {(0, \ldots, 0,1)^{\rm T}}$. So without loss of generality, we may assume that $\beta = {(0, \ldots, 0,1)^{\rm T}}$. Then, ${J\cdot\lambda  - \beta}$ has $n - 1$ equations in $\left\langle {{{\{ {x_1}, \ldots, {x_n},1\} }^d } \otimes \{ {\lambda _1}, \ldots, {\lambda _k}\} } \right\rangle $ and one equation in $\left\langle {{{\{ {x_1}, \ldots, {x_n},1\} }^d } \otimes \{ {\lambda _1}, \ldots, {\lambda _k},1\} } \right\rangle $.
\end{remark}

\begin{theorem}\label{th:1} Let $F(x,\lambda ) = \left\{ {f,J \cdot \lambda  - \beta} \right\}$ be given as in (\ref{eq:2}), $\beta = {(0, \ldots, 0,1)^{\rm T}}$, and $G = \left\{ {f,g} \right\}$ where $g = \left\{ {g_1}, \ldots, {g_n}\right\} $. ${g_i} = {l_{i1}} \cdots {l_{id }}{h_i} \in {\left\langle {{x_1}, \ldots, {x_n},1} \right\rangle ^d } \times \left\langle {{\lambda _1}, \ldots, {\lambda _k}} \right\rangle $ for $i = 1, \ldots, n - 1$ ; where ${l_{ij}}$ are linear functions in $\bbbc[{x_1}, \ldots, {x_n}]$, ${h_i}$ with randomly chosen coefficient and ${h_i}$ are homogeneous linear functions in $\bbbc[{\lambda _1}, \ldots, {\lambda _k}]$, $i = 1, \ldots, n - 1$, $j = 1, \ldots, d$ and ${g_n} = \sum\limits_{i = 1}^k {{\lambda _i}} {g_{ni}} - 1$. $H:{\bbbc^n} \times {\bbbc^k} \times \bbbc \to {\bbbc^{n + k}}$ be the homotopy defined by $H(x,\lambda ,t) = G \cdot (1 - t) + F \cdot \gamma  \cdot t$ where $\gamma $ is a randomly chosen complex number for Gamma Trick (see \cite{Sommese2005The} Chapter 7 for details). Then, generically the following items hold,
\begin{enumerate}
\item The set $S \subseteq {\bbbc^{n + k}}$ of roots of $H(x,\lambda ,0) = G(x,\lambda )$ is finite and each is a nonsingular solution of $H(x,\lambda ,0)$.
\item The number of points in $S$ is equal to the maximum number of isolated solutions of $H(x,\lambda ,0)$ as coefficients of ${l_{ij}}$, ${h_i}$, ($i = 1, \ldots, n - 1$, $j = 1, \ldots, d$) and $\gamma$ vary over $\bbbc$.
\item The solution paths defined by $H$ starting, with $t=0$, at the points in $S$ are trackable.
\end{enumerate}

\end{theorem}

\begin{proof}
As for item 1, since $f$ has $k$ equations only in $x$, and $V({f_1}, \ldots, {f_k})$ is a pure $n-k$ dimension algebraic set in ${\bbbc^n}$. To solve system $G$, it needs only $n - k$ linear functions $L$ in $g$ from different ${g_i}$ with $i \in \left\{ {1, \ldots, n - 1} \right\}$ to determine $x$. $\left\{ {f,L} \right\}$ is a $n \times n$ square system, $V(f,L)$ is a finite witness set for algebraic set $V({f_1}, \ldots, {f_k})$, and each of the points is a nonsingular solution of $V(f,L)$ (see \cite{Sommese2005The} Chapter 13 for details). And, we finally determine $\lambda$ by solving a square linear equations.
As for item 2, and item 3, it's a trivial deduction of Coefficient-Parameter Continuation \cite{MORGAN1989123}.\qed
\end{proof}

\begin{remark}\label{re:3}
From the proof of Theorem \ref{th:1}, the number of points of the finite set $V(f,L)$ is the degree of $V(f)$, and is independent of the choice of $L$. Thus, based on the number of different choices of $L$, and item 2, we can give a root count bound of system $F(x,\lambda )$ as in the following theorem, which is similar to the bound in \cite{Wu2017}.
\end{remark}

\begin{theorem}\label{th:2}
For a system $F(x,\lambda ) = \left\{ {f,J \cdot \lambda  - \beta} \right\}$ as in (\ref{eq:2}). The number of complex root of this system is bounded by
\begin{equation}
 \left( {\begin{array}{*{20}{c}}
  {n - 1} \\
  {n - k}
\end{array}} \right){d^{n - k}}D
\end{equation}
\\
where $D$ is the degree of $V(f)$.
\end{theorem}

Due to Theorem \ref{th:1}, its proof and the remarks, we can design an efficient procedure to numerically find the isolated solutions of system $F(x,\lambda ) = \left\{ {f,J \cdot \lambda  - \beta} \right\}$ in the form of (\ref{eq:2}). First, we solve a square systems $\left\{ {f,L} \right\}$, where $L$ are $n-k$ randomly generated linear functions . Then for each group of $n-k$ linear functions $L'$ chosen in $g$ from different ${g_i}$ with $i \in \left\{ {1, \ldots, n - 1} \right\}$, we construct linear homotopy from $\left\{ {f,L} \right\}$ to $\left\{ {f,L'} \right\}$, starting from points of  $V(f,L)$, and solve the square linear equation of $\lambda$ respectively. Let $S$ be the set consist of all the pairs of $x$ and $\lambda$, {\it i.e.} $(x,\lambda)$. Finally construct linear homotopy $H(x,\lambda ,t) = G \cdot (1 - t) + F \cdot \gamma  \cdot t$ starting from points in $S$, thus the endpoints of the convergent paths of homotopy $H(x,\lambda ,t)$ are isolated solutions of system $F(x,\lambda ) = \left\{ {f,J \cdot \lambda  - \beta} \right\}$. We put specific description of this procedure in the next section.

\section{Algorithm}

From Theorem \ref{th:1}, its proof and Remarks \ref{re:2} \& \ref{re:3}, we propose an approach for computing isolated solutions of system $F(x,\lambda ) $ as described in the end of last section. For consideration of the sparsity, we use the polyhedral homotopy method for solutions of the square system $\left\{ {f,L} \right\}$. Actually, we use polyhedral homotopy method only once. Now we describe our algorithms.

\begin{algorithm}[H]\label{alg:1}
\SetKwInOut{Input}{input}\SetKwInOut{Output}{output}
\Input{  $(n + k) \times (n + k)$ square polynomial system $F(x,\lambda ) = \left\{ {f,J \cdot \lambda  - \beta} \right\}$ as in (\ref{eq:2});}
\Output{  finite subset $V(F)$ of ${\bbbc^{n + k}}$ }
\SetAlgoLined
\BlankLine
 Let $L = \{ {l_1}, \ldots, {l_{n - k}}\} $ where ${l_i}$ are linear equations with randomly chosen coefficients in $\bbbc$\;
 Solve system $ \left\{ {f,l} \right\}$ by polyhedral homotopy method and denote the solution set as $M$\;
 Let $F'(x,\lambda ) = \left\{ {f,A\cdot(J\cdot\lambda  - \beta)} \right\}$, $G = \left\{ {f,g} \right\}$, $A \in G{L_n}(\bbbc)$ such that $A \cdot \beta = (0, \ldots, 0,1)$, $g = \{ {g_1}, \ldots, {g_n}\} $. ${g_i} = {l_{i1}} \cdots {l_{id }} \cdot {h_i} \in {\left\langle {{x_1}, \ldots, {x_n},1} \right\rangle ^{d }} \times \left\langle {{\lambda _1}, \ldots, {\lambda _k}} \right\rangle $ for $i = 1, \ldots, n - 1$ with coefficients randomly chosen in $\bbbc$, and ${g_n}$ is the last equation of ${A\cdot(J\cdot\lambda  - \beta)}$\;
 Let $C = \left\{ {I\left| {I = ({\alpha _1}, \ldots, {\alpha _{n - 1}}) \in {{\left\{ {0,1} \right\}}^{n - 1}},\sum\limits_{i = 1}^{n - 1} {{\alpha _i}}  = n - k} \right.} \right\}$, and $\Omega  = \emptyset $\;
 \Repeat{$C=\emptyset$}{
      Pick one vector $I = ({\alpha _1}, \ldots, {\alpha _{n-1}})$ from $C$, and $C = C\backslash I$\;
      Let $L' = \emptyset $\;
          \For{ $i$ from $1$ to $n-1$ }{
              \If {${\alpha _i} = 1$}{ pick one linear equation ${l_i}'$ from $\{ {l_{i1}}, \ldots, {l_{id}}\} $ and $L' = L' \cup \{ {l_i}'\} $.
              }
              }
Construct linear homotopy ${H_1}(x,t) = \left\{ {f,L} \right\}\cdot (1 - t) + \left\{ {f,L'} \right\} \cdot {\gamma _1} \cdot t$ starting at points in $M$. ${\gamma _1}$ is randomly chosen complex number for gamma trick. Let the set of endpoints of the tracked paths be $M'$\;
Take every point ${x^*} = (x_{_1}^*, \ldots, x_{_n}^*)$ in $M'$ into the system $G = \left\{ {f,g} \right\}$ and resolve ${\lambda ^*} = (\lambda _1^*, \ldots, \lambda _k^*)$. $\Omega  = \Omega  \cup \{ ({x^*},{\lambda ^*})\} $. \;}
Construct linear homotopy ${H_2}(x,\lambda ,t) = G \cdot (1 - t) + F \cdot {\gamma _2} \cdot t$ starting at points in $\Omega $, ${\gamma _2}$ is randomly chosen complex number for gamma trick. Let the set of convergent endpoints of the tracked paths be $ V(F)$\;
\Return{ $V(F)$}\;
\caption{\texttt{LPH} (Linear Product Homotopy)}
\end{algorithm}

\begin{remark}\label{re:4}
$\# C = \left( {\begin{array}{*{20}{c}}
  {n - 1} \\
  {n - k}
\end{array}} \right)$ , and  in Step 5, $I = ({\alpha _1}, \ldots, {\alpha _{n-1}})$ has and only has $n - k$ entries ${\alpha _i} = 1$. When ${\alpha _i} = 1$, we choose linear equation in ${g_i}$, and there are $d$ candidates $\{ {l_{i1}}, \ldots, {l_{id}}\} $ to choose. It adds up to be $\left( {\begin{array}{*{20}{c}}
  {n - 1} \\
  {n - k}
\end{array}} \right){d ^{n - k}}$ different $\left\{ {f,L'} \right\}$. Each $\left\{ {f,L'} \right\}$ has the same number $D=\deg(V(f))$ of isolated roots as $ \left\{ {f,L} \right\}$, so homotopy in Step 13 will have no path divergent. Thus we have $\left( {\begin{array}{*{20}{c}}
  {n - 1} \\
  {n - k}
\end{array}} \right){d^{n - k}}D$ points in $\Omega $, which is the root bound we mention in Theorem \ref{th:2}. It would happen that some of the homotopy paths divergent in Step 16, the method of end games for homotopy should be used \cite{Morgan1992}\cite{MORGAN198677}\cite{Huber1998}\cite{IOPORT.05995255}.
\end{remark}

\section{Real Critical Set}
In this section, we will combine the {\tt LPH} algorithm in Section 4 and methods in \cite{Wu:2013:FPR:2465506.2465954} to compute a real witness set which has at least one point on each irreducible component of a real algebraic set, and give an illustrative example.

\subsection{Critical Points on a Real Algebraic Set}
We make the following assumptions (adapted from \cite{Wu:2013:FPR:2465506.2465954}). Let $f:{\bbbc^n} \to {\bbbc^k}$ be a polynomial system, and $f = ({f_1}, \ldots, {f_k})$ in $\bbbr[{x_1}, \ldots, {x_n}]$ satisfying the so-called Full Rank Assumption:
\begin{enumerate}\label{A1}
\item ${V_\bbbr}({f_1}, \ldots, {f_i})$ has dimension $n - i$ for $i = 1, \ldots,  k$;
\item the ideal $I({f_1}, \ldots, {f_i})$ is radical for $i = 1, \ldots, k$.
\end{enumerate}

Under these assumptions, $(\nabla f_1^T, \ldots, \nabla f_i^T)$ has rank $i$ for a generic point $p \in V({f_1}, \ldots, {f_i})$ for $i = 1, \ldots, k$.

The main problem we consider is finding at least one real witness point on each real  dimensional components of ${V_\bbbr}(f)$. For this purpose, we choose $\Phi $ in Definition 3 to be a linear function with $\Phi  = x \cdot \beta + c$, where $\beta$ is a random vector in ${\bbbr^n}$, and $c$ is a random real number. Then system (\ref{eq:1}) becomes
\begin{equation}\label{eq:F}
F = \left\{ {f,\sum\limits_{i = 1}^k {{\lambda _i}\nabla {f_i}}  - \beta} \right\} = 0.
\end{equation}

It may happen that there is no critical points of $\Phi $ in some connected component of ${V_\bbbr}({f_1}, \ldots, {f_k})$. In that case, we add $\Phi $ to $f$ and construct a system with $k + 1$ equations
\begin{equation}\label{eq:f1}
 {f^{(1)}} = \left\{ {f,x \cdot \beta + c} \right\}.
\end{equation}
Then recursively, we choose another linear function $\Phi_1 $, compute the critical points of $\Phi_1 $ with respect to $V({f^{(1)}})$; and so on.

We give a concrete definition of the set of real witness points ${W_\bbbr}(f)$ we are going to compute (see \cite{Wu:2013:FPR:2465506.2465954}).

\begin{definition} Let $f:{\bbbc^n} \to {\bbbc^k}$  be a polynomial system, $k \leqslant n$, and $f = ({f_1}, \ldots, {f_k})$ in $\bbbr[{x_1}, \ldots, {x_n}]$ satisfying Full Rank Assumption. $F$ and ${f^{(1)}}$ defined as (\ref{eq:F}) and (\ref{eq:f1}). We define ${W_\bbbr}(f)$ as follows:

\begin{enumerate}
\item ${W_\bbbr}(f) = {V_\bbbr}(f)$ if $n = k$;
\item ${W_\bbbr}(f) = {V_\bbbr}(F) \cup {W_\bbbr}({f^{(1)}})$ if $k < n$.
\end{enumerate}
\end{definition}

It is obvious from the definition that we can recursively solve the square system (\ref{eq:F}), and apply plane distance critical points formulation of ${f^{(1)}}$ to finally get the set of witness points ${W_\bbbr}(f)$ which contains finitely many real points on ${V_\bbbr}(f)$, and there is at least one point on each connected component of ${V_\bbbr}(f)$. Since the formulation introduces auxiliary unknowns, it increases the size of the system and leads to computational difficulties. For example, when $n = 15$ and $k = 10$, the size of system (\ref{eq:F}) becomes $25$, which is challenging for general homotopy software. Combining the {\tt LPH} algorithm, Theorem \ref{th:1} and Theorem \ref{th:2}, we have the following algorithm and an upper bound of number of points in ${W_\bbbr}(f)$, as in \cite{Wu2017}.

\begin{algorithm}[H]\label{alg:2}
\SetKwInOut{Input}{input}\SetKwInOut{Output}{output}
\Input{a polynomial system $f = ({f_1}, \ldots, {f_k})$, $k \leqslant n$, which satisfy the full rank assumption;}
\Output{a finite subset ${W_\bbbr}(f)$ of ${\bbbr^n}$, which contains at least one point on each connected component of the real algebraic set ${V_\bbbr}(f)$}
\SetAlgoLined
\BlankLine
Let ${W_\bbbr}(f) = \emptyset $\;
\While{ $k \leqslant n$}{
 ${V_\bbbr} \leftarrow \texttt{LPH} (f,\Jac_f(x) \cdot \lambda  - \beta)$\;
 ${W_\bbbr}(f) \leftarrow {W_\bbbr}(f) \cup {V_\bbbr}$\;
 $f \leftarrow \left\{ {f,x \cdot \beta + c} \right\}$ where $n$ is a random vector in ${\bbbr^n}$, and $c$ is a random real number\;
 $k \leftarrow k + 1$\;
 }
 \caption{\texttt{RWS} (Real Witness Set)}
 \Return{ ${W_\bbbr}(f)$}
\end{algorithm}
\begin{remark}
Algorithm \ref{alg:2} is essentially a recursive calling of Algorithm \ref{alg:1}.
\end{remark}

\begin{theorem}
(\cite{Wu2017} Theorem 2.1) For a system $f = ({f_1}, \ldots, {f_k})$ with $n$ variables and degrees ${d_i} = \deg ({f_i})$ for $i = 1, \ldots, k$. The number of complex root of system (\ref{eq:F}) is bounded by
\begin{equation}
 \left( {\begin{array}{*{20}{c}}
  {n - 1} \\
  {n - k}
\end{array}} \right){(d - 1)^{n - k}}D
\end{equation}
\\
where $d = \max \{ {d_1}, \ldots, {d_k}\}  > 1$ and $n > k > 0$, $D$ is the degree of the pure $n - k$ dimensional component of $V = V(f)$.
\\
Moreover, the total number of points in ${W_\bbbr}(f)$  is bounded by
\begin{equation}
\sum\limits_{j = 0}^{n - k} {\left( {\begin{array}{*{20}{c}}
  {n - 1 - j} \\
  {n - k - j}
\end{array}} \right){{(d - 1)}^{n - k - j}}D}.
\end{equation}
\end{theorem}


Obviously we have the following inequalities: \[MV(F) \leqslant \left( {\begin{array}{*{20}{c}}
  {n - 1} \\
  {k - 1}
\end{array}} \right){(d - 1)^{n - k}}D \leqslant \left( {\begin{array}{*{20}{c}}
  {n - 1} \\
  {k - 1}
\end{array}} \right){(d - 1)^{n - k}}\prod\limits_{i = 1}^k {{d_i}}  \leqslant {d^n}\prod\limits_{i = 1}^k {{d_i}}. \]
If $f$ is dense, the equalities hold.  And if $f$ is sparse, they vary considerably most of the time. For example, let
$f = \{  - 62xy + 97y - 4xyz - 4,80x - 44xy + 71{y^2} - 17{y^3} + 2\}$ with $d = 3,n = 3,k = 2$. We have $MV(F) = 11$, $\left( {\begin{array}{*{20}{c}}
  {n - 1} \\
  {k - 1}
\end{array}} \right){(d - 1)^{n - k}}D = 28$, $\left( {\begin{array}{*{20}{c}}
  {n - 1} \\
  {k - 1}
\end{array}} \right){(d - 1)^{n - k}}\prod\limits_{i = 1}^k {{d_i}}  = 36$, and ${d^n}\prod\limits_{i = 1}^k {{d_i}}=243 $.

\subsection{Illustrative Example}
In this subsection, we present an illustrative example for Algorithm \ref{alg:2}.
\begin{example} \label{Cubic and Cubic ellipse}
 Consider the hypersurface defined by $f = ({y^2} - {x^3} - ax - b) \cdot ({(x - y + e)^3} + x + y)$, $e = 6,a =  - 4,b =  - 1$. Clearly, ${V_\bbbr}(f)$ is the combination of a cubic ellipse $({y^2} - {x^3} - ax - b)$, and a cubic curve ${(x - y + e)^3} + x + y$, as plotted in Fig. 1.
We show how to compute ${W_\bbbr}(f)$ by Algorithm \ref{alg:2}.

\begin{figure}\label{fig:1}
\begin{center}
    \includegraphics[width=7cm]{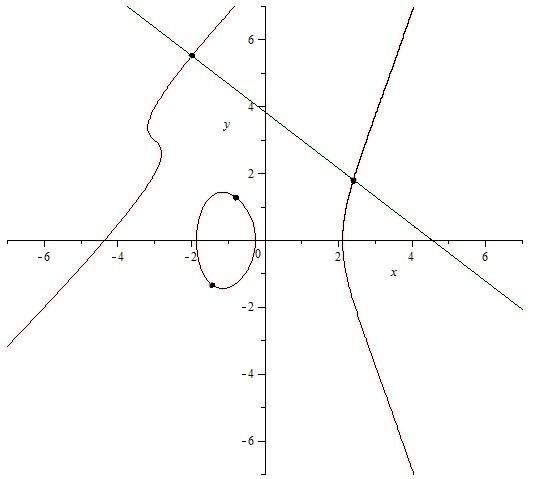}
\end{center}
  \caption[]{$n=10,k=4,\deg=2$}
\end{figure}

\begin{itemize}
\item For computing ${V_\bbbr} = \texttt{LPH}(f)$, we randomly choose a line $l$ in ${\bbbc^2}$ and solve $L = \left\{ {f,l} \right\}$ by polyhedral homotopy, which follows $D = 6$ paths. Then to compute $\Omega $ by linear homotopy, we follow $\left( {\begin{array}{*{20}{c}}
  {2 - 1} \\
  {2 - 1}
\end{array}} \right){(6 - 1)^{2 - 1}}6 = 30$ convergent paths, and for ${V_\bbbr}$ by linear homotopy, we follow 30 paths, of which 6 are convergent and 19 divergent. Then \[{V_\bbbr} = \left\{ {\begin{array}{*{20}{c}}
  {( - 1.44299, - 1.32941), }&{( - 0.781143,1.28371)}
\end{array}} \right\}\]
\item For computing ${W_\bbbr}(f)$, we solve ${f^{(1)}} = \left\{ {f,x \cdot \beta + c} \right\}$ by polyhedral homotopy, with $x \cdot \beta + c = 0.874645x + 1.0351y - 3.9825$ and \[{W_\bbbr}({f^{(1)}}) = \left\{ {\begin{array}{*{20}{c}}
  {(2.4052801,1.815026), }&{( - 1.992641,5.531208)}.
\end{array}} \right\}\]
So ${W_\bbbr}(f) = {W_\bbbr}({f^{(1)}}) \cup {V_\bbbr}$, which has at least one point in each connected component of ${V_\bbbr}(f)$ as in Fig. 1.
\end{itemize}
\end{example}

\section{Experiment Performance}
 As shown in Section 5, to compute the set ${W_\bbbr}(f)$, the key and most time consuming steps are solving the system $F = \left\{ {f,\sum\limits_{i = 1}^k {{\lambda _i}\nabla {f_i}}  - \beta} \right\}$ in Algorithm \ref{alg:1}. In this section, given $f = \left\{ {{f_1}, \ldots, {f_k}} \right\}$, we solve the square system $F = \left\{ {f,\sum\limits_{i = 1}^k {{\lambda _i}\nabla {f_i}}  - \beta} \right\}$. We compare our program \texttt{LPH} which implements Algorithm \ref{alg:1} to Hom4PS-2.0 (available at http://www.math.msu.edu/~li).
 All the examples were computed on a PC with Intel Core i5 processor (2.5GHz CPU, 4 Cores and 6 GB RAM) in the Windows environment. 
 We mention that \texttt{LPH} is a program written in C++, available at http://arcnl.org/PDF/LHP.zip, and an interface of Maple is provided on this site.
\subsection{Dense Examples}

\begin{table}
\caption{Dense Examples}
\begin{minipage}{0.25\linewidth}
\begin{tabular}{|c|c|c|c|c|}
\hline
$n$ & $k$ & T1 & T2  & RAT\\
\hline
2 & 1 & 0.125s & 0.094s  & 1.32\\
\hline
3 & 1 & 0.125s & 0.109s  & 1.14\\
\hline
3 & 2 & 0.125s & 0.109s  & 1.14\\
\hline
4 & 1 & 0.125s & 0.109s  & 1.14\\
\hline
4 & 2 & 0.156s & 0.156s  & 1.00\\
\hline
4 & 3 & 0.202s & 0.265s  & 0.76\\
\hline
5 & 1 & 0.125s & 0.109s  & 1.14\\
\hline
5 & 2 & 0.187s & 0.202s & 0.93\\
\hline
5 & 3 & 0.390s & 0.655s  & 0.59\\
\hline
5 & 4 & 0.687s & 1.280s  & 0.54\\
\hline
6 & 1 & 0.140s & 0.109s & 1.28\\
\hline
6 & 2 & 0.281s & 0.328s  & 0.86\\
\hline
6 & 3 & 0.76s & 1.68s & 0.45\\
\hline
6 & 4 & 1.61s & 4.36s  & 0.37\\
\hline
6 & 5 & 1.90s & 6.59s & 0.29\\
\hline
7 & 1 & 0.14s & 0.10s  & 1.29\\
\hline
7 & 2 & 0.344s & 0.56s  & 0.61\\
\hline
7 & 3 & 1.3s & 3.82s  & 0.34\\
\hline
7 & 4 & 3.7s & 14.1s  & 0.27\\
\hline
7 & 5 & 6.318s & 27.9s  & 0.23\\
\hline
7 & 6 & 6.006s & 27.6s  & 0.21\\
\hline
8 & 1 & 0.15s & 0.15s  & 1.00\\
\hline
8 & 2 & 0.54s & 0.73s  & 0.74\\
\hline
8 & 3 & 2.29s & 7.2s  & 0.31\\
\hline
8 & 4 & 8.018s & 34.1s  & 0.23\\
\hline
8 & 5 & 17.6s & 91.2s  & 0.19\\
\hline
8 & 6 & 23.7s & 153s  & 0.15\\
\hline
8 & 7 & 19.4s & 128s  & 0.15\\
\hline
9 & 1 & 0.2s & 0.18s  & 1.08\\
\hline
9 & 2 & 0.7s & 1.2s & 0.58\\
\hline
9 & 3 & 4.1s & 13.1s  & 0.31\\
\hline
9 & 4 & 16.1s & 1m19s  & 0.20\\
\hline
9 & 5 & 46.5s & 4m29s  & 0.17\\
\hline
9 & 6 & 1m20s & 9m38s & 0.138\\
\hline
9 & 7 & 1m30s & 11m10s  & 0.135\\
\hline
9 & 8 & 59.9s & 7m52s  & 0.126\\
\hline
10 & 1 & 0.23s & 0.18s  & 1.25\\
\hline
10 & 2 & 0.98s & 1.9s  & 0.51\\
\hline
10 & 3 & 5.8s & 24.8s & 0.24\\
\hline
10 & 4 & 31.1s & 2m53s  & 0.18\\
\hline
10 & 5 & 1m46s & 11m30s  & 0.15\\
\hline
{\;10\;} & {\;6\;} & {\quad 3m41s \quad} & 29m15s  & 0.13\\
\hline
10 & 7 & 5m10s & 48m32s  & 0.107\\
\hline
10 & 8 & 4m57s & 48m27s  & 0.102\\
\hline
10 & 9 &  3m2s & \quad 30m7.553s \quad & 0.1\\
\hline
\end{tabular}
\end{minipage}
\hfill
\begin{minipage}{0.45\linewidth}
\begin{tabular}{|c|c|c|c|c|}
\hline
$n$ & $k$ & T1 & T2  & RAT\\
\hline
11 & 1 & 0.28s & 0.23s  & 1.37\\
\hline
11 & 2 & 1.20s & 2.85s  & 0.42\\
\hline
11 & 3 & 8.7s & 41.5s  & 0.209\\
\hline
11 & 4 & 51.5s & 5m25s  & 0.158 \\
\hline
11 & 5 & 3m15s & 25m30s  & 0.128 \\
\hline
11 & 6 & 9m1s & 76m32s  & 0.118 \\
\hline
11 & 7 & 16m14s & 2h.45m30s  & 0.098 \\
\hline
11 & 8 & 18m34s & 3h52m52s  & 0.079 \\
\hline
11 & 9 & 16m22s & 3h34m38s  & 0.076 \\
\hline
11 & 10 & 6m37s & overflow & $\varepsilon $ \\
\hline
12 & 1 & 0.29s & 0.218s &  1.360 \\
\hline
12 & 2 & 1.27s & 4.3s &  0.294 \\
\hline
12 & 3 & 13.5s & 1m10s &  0.191 \\
\hline
12 & 4 & 1m25s & 10m0.2s & 0.142 \\
\hline
12 & 5 & 6m17s & 53m58s  & 0.116 \\
\hline
12 & 6 & 20m1s & 3h15m10s  & 0.102\\
\hline
12 & 7 & 44m23s & 8h7m1s  & 0.091 \\
\hline
12 & 8 & 1h8m30s & 14h38m24s  & 0.0779 \\
\hline
12 & 9 & 1h8m42s & overflow & $\varepsilon $ \\
\hline
12 & 10 & 46m21s & overflow &  $\varepsilon $ \\
\hline
12 & 11 & 21m54s & overflow  & $\varepsilon $ \\
\hline
13 & 1 & 0.343s & 0.218s  & 1.573 \\
\hline
13 & 2 & 1.716s & 6.193s  & 0.277\\
\hline
13 & 3 & 18s & 1m51s  & 0.61 \\
\hline
13 & 4 & 2m15s & 18m35s  & 0.121 \\
\hline
13 & 5 & 11m10s & 1h52m27s  & 0.099\\
\hline
13 & 6 & 40m6s & 7h16m34s  & 0.092 \\
\hline
13 & 7 & 1h39m40s & 21h25m14s  & 0.078 \\
\hline
13 & 8 & 2h58m48s & overflow & $\varepsilon $ \\
\hline
13 & 9 & 3h59m32s & overflow & $\varepsilon $ \\
\hline
13 & 10 & 3h40m3s & overflow & $\varepsilon $\\
\hline
13 & 11 & 2h13m9s & overflow & $\varepsilon $ \\
\hline
13 & 12 & {\;56m48.309s\;} & overflow  & $\varepsilon $ \\
\hline
14 & 2 & 2.5s & 9.6s  & 0.264 \\
\hline
14 & 3 & 24.3s & 3m0.4s  & 0.134 \\
\hline
14 & 4 & 3m19s & 37m19s  & 0.089 \\
\hline
14 & 5 & 19m28s & 8h24m29s & 0.038 \\
\hline
14 & 6 & 1h16m20s & {\;15h58m59s\;}  & 0.079 \\
\hline
14 & 7 & 3h34m52s & overflow & $\varepsilon $ \\
\hline
14 & 8 & 7h50m34s & overflow & $\varepsilon $ \\
\hline
14 & 9 & 12h43m8s & overflow &  $\varepsilon $ \\
\hline
14 & 10 & 16h48m4s & overflow &  $\varepsilon $\\
\hline
14 & 11 & 13h9m8s & overflow &  $\varepsilon $ \\
\hline
14 & 12 & 6h29m37s & overflow &  $\varepsilon $ \\
\hline
{\;14\;} & {\;13\;} & 2h18m27s & overflow&  $\varepsilon $ \\

\hline
\end{tabular}
\end{minipage}
\end{table}

In Table 1, we provide the timings of {\tt LPH} and Hom4ps-2.0 for solving systems $F = \left\{ {f,\sum\limits_{i = 1}^k {{\lambda _i}\nabla {f_i}}  - \beta} \right\}$, where $f = ({f_1}, \ldots, {f_k})$ consists of dense polynomials of degree 2, $n = 2,...,14$ and $1 \leqslant k \leqslant n - 1$. T1 ,T2 are the the timings for {\tt LPH} and Hom4ps-2.0, respectively, and RAT is the ratio of T1 to T2. 
`overflow' means running out of memory. When T2=overflow, we set RAT=$\varepsilon$.

It may be observed that {\tt LPH} is much faster than Hom4ps-2.0 when $k > 1$. Note also that {\tt LPH} is a little bit slower than Hom4ps-2.0 when $k=1$. The main reason is obvious. 
That is, the root number bound of {\tt LPH}, {\it i.e.}
$ \left( {\begin{array}{*{20}{c}}
  {n - 1} \\
  {n - k}
\end{array}} \right){(d - 1)^{n - k}}D,$
is close to the mixed volume $MV(F)$ when $F$ is dense but the computation of $MV(F)$ is very time-consuming.

\subsection{Sparse Examples}
\begin{table}
\caption{Sparse Examples}
\begin{tabular}{|c|c|c|c|c|c|c|c|c|c|c|}
\hline
Ex & $n$ & $k$ & $d$ & term &\#1 & \#2 & \# & T1 & T2 & RAT\\
\hline
C2 & 5 & 4 & 4-5 & 7-32 & 2*1767 & 692 & 383 & 16.6s & 19.4s & 0.85\\
\hline
M3 & 9 & 5 & 2 & 2-7 & 2*2240 & 368 & 32 & 69s & 6s & 11\\
\hline
G2 & 5 & 2 & 4 & 8-9 & 2*1080 & 17 & 15 & 8.4s & 0.2s & 31.8\\
\hline
H1 & 8 & 6 & 1-3 & 2-4 & 2*400 & 15 & 15 & 9.8s & 0.18s & 52\\
\hline
H2 & 8 & 5 & 2-4 & 3-5 & 2*12320 & 148 & 80 & 5m49s & 1.2s & 267\\
\hline
\end{tabular}
\end{table}

In Table 2, we provide the timings of {\tt LPH} and Hom4ps-2.0 on sparse examples: Czapor Geddes2, Morgenstern AS(3or), Gerdt2, Hairer1, and Hawes2 which are available at : http://www-sop.inria.fr/saga/POL/. \#1 and \#2 is the number of curves followed by {\tt LPH} and Hom4ps-2.0, respectively. \# is the number of roots of the Jacobian systems constructed from the examples. `$d$' means the minimal and maximal degree of the example. ``term'' means the minimal and maximal number of terms of the example. T1 and T2 are the timings of {\tt LPH} and Hom4ps-2.0, respectively. RAT means the ratio of T1 to T2.

Note that {\tt LPH} is much slower than Hom4ps on these sparse examples. The main reason is that {\tt LPH} pays the overhead cost for the $\Omega $ and homotopy $${H_2}(x,\lambda ,t) = G \cdot (1 - t) + F \cdot {\gamma _2} \cdot t.$$ Moreover, {\tt LPH} executes $2*\left( {\begin{array}{*{20}{c}}
  {n - 1} \\
  {k - 1}
\end{array}} \right){(d - 1)^{n - k}}D$ times of curve following, while Hom4ps does only $MV(F)$ times of curve following. $\left( {\begin{array}{*{20}{c}}
  {n - 1} \\
  {k - 1}
\end{array}} \right){(d - 1)^{n - k}}D$ is not tight for these sparse examples and much greater than $MV(F)$.
\subsection{RAT/Density}
\begin{figure}
\begin{center}
    \includegraphics[width=10cm]{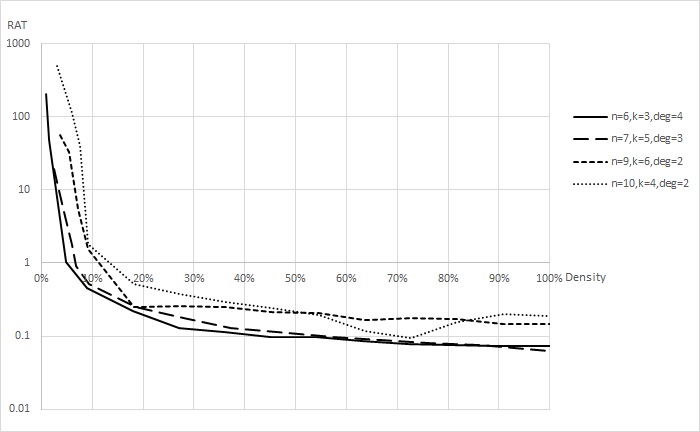}
\end{center}
  \caption[]{RAT and Density}
\end{figure}

In Fig. 2, we present the changes of ratio of T1 to T2 as terms increase. We randomly generate $f = ({f_1}, \ldots, {f_k})$ with different $n, k$ and degrees, and increase the number of terms from 2 to dense.

It can be observed that, when the polynomials are not very sparse, {\it e.g.} the number of terms are more than $10\% $ of $\left( {\begin{array}{*{20}{c}}
  {n + d} \\
  {d}
\end{array}} \right)$, {\tt LPH} is faster than Hom4ps-2.0. Actually, when the polynomials are not very sparse, the root number bound $\left( {\begin{array}{*{20}{c}}
  {n - 1} \\
  {k - 1}
\end{array}} \right){(d - 1)^{n - k}}D$ is close to $MV(F)$.
\section{Acknowledgement}

We gratefully acknowledge the very helpful suggestions of Hoon Hong on this paper with emphasize on Section 6. We also thank Changbo Chen for his helpful comments.

\bibliography{mybibfile}

\end{document}